\documentclass{msc}

\usepackage{amsmath}
\usepackage{graphicx}
\usepackage{xspace}
\usepackage{enumerate}
\usepackage[usenames,dvipsnames]{xcolor}
\definecolor{Revolutionary}{RGB}{232,70,68}
\usepackage{cleveref}
\usepackage[utf8]{inputenc}
\usepackage{color}
\usepackage{hyperref}
\usepackage{url}

\usepackage{amsfonts}
\usepackage{dutchcal}
\usepackage{amsthm}
\usepackage{cases}
\usepackage{mdframed}
\usepackage{stmaryrd}
\usepackage{mathpartir}
\usepackage{listings}
\usepackage{tikz}
\usepackage{tikz-cd}

\usepackage{utf8}
\usepackage{agdahen}

\usepackage{scalerel}
\usepackage{stackengine,wasysym}

\newcommand\reallywidetilde[1]{\ThisStyle{%
  \setbox0=\hbox{$\SavedStyle#1$}%
  \stackengine{-.1\LMpt}{$\SavedStyle#1$}{%
    \stretchto{\scaleto{\SavedStyle\mkern.2mu\AC}{.5150\wd0}}{.6\ht0}%
  }{O}{c}{F}{T}{S}%
}}

\newcommand{\teletype}[1]{\ensuremath{\mathtt{#1}}}
\newcommand{\systemname}[1]{\teletype{\color{darkgray}#1}\xspace}

\newcommand{\Agda}{\systemname{Agda}}

\newcommand{\agdaunimath}{\systemname{agda}-\systemname{unimath}}

\newcommand{\Coq}{\systemname{Coq}}
\newcommand{\CubicalAgda}{\systemname{Cubical} \systemname{Agda}}

\definecolor{Revolutionary}{RGB}{232,70,68}

\newcommand{\UniMath}{\systemname{UniMath}}

\newcommand{\ie}{{i.e.}}

\newcommand{\defeq}{\mathrel{:=}}

\newcommand{\agdalink}[1]{\href{#1}{\AgdaHen}}
\usepackage{catchfilebetweentags}
\usepackage{newunicodechar}
\newunicodechar{≃}{\ensuremath{\simeq}}
\newunicodechar{≅}{\ensuremath{\cong}}
\newunicodechar{∈}{\ensuremath{\in}}
\newunicodechar{⁻}{\ensuremath{^{-}}}
\newunicodechar{ᴹ}{\ensuremath{_{M}}}
\newunicodechar{ᴺ}{\ensuremath{_{N}}}
\newunicodechar{≡}{\ensuremath{\equiv}}
\newunicodechar{λ}{\ensuremath{\lambda}}
\newunicodechar{⊎}{\ensuremath{\uplus}}
\newunicodechar{∷}{\ensuremath{::}}
\newunicodechar{ℓ}{\ensuremath{\ell}}
\newunicodechar{ᵢ}{\ensuremath{_i}}
\newunicodechar{⟨}{\ensuremath{\langle}}
\newunicodechar{⟩}{\ensuremath{\rangle}}
\newunicodechar{α}{\ensuremath{\alpha}}
\newunicodechar{β}{\ensuremath{\beta}}
\newunicodechar{θ}{\ensuremath{\theta}}
\newunicodechar{φ}{\ensuremath{\varphi}}
\newunicodechar{ψ}{\ensuremath{\psi}}
\newunicodechar{η}{\ensuremath{\eta}}
\newunicodechar{ε}{\ensuremath{\varepsilon}}
\newunicodechar{ι}{\ensuremath{\iota}}
\newunicodechar{Σ}{\ensuremath{\Sigma}}
\newunicodechar{σ}{\ensuremath{\sigma}}
\newunicodechar{∀}{\ensuremath{\forall}}
\newunicodechar{ℕ}{\ensuremath{\mathbb{N}}}
\newunicodechar{→}{\ensuremath{\to}}
\newunicodechar{⊎}{\ensuremath{\uplus}}
\newunicodechar{⋆}{\ensuremath{*}}
\newunicodechar{¬}{\ensuremath{\lnot}}
\newunicodechar{Θ}{\ensuremath{\theta}}
\newunicodechar{∧}{\ensuremath{\wedge}}
\newunicodechar{∨}{\ensuremath{\vee}}
\newunicodechar{∼}{\ensuremath{\sim}}
\newunicodechar{≢}{\ensuremath{\nequiv}}
\newunicodechar{Π}{\ensuremath{\Pi}}
\newunicodechar{ℤ}{\ensuremath{\mathbb{Z}}}
\newunicodechar{∥}{\ensuremath{\parallel}}
\newunicodechar{∣}{\ensuremath{\mid}}
\newunicodechar{ℚ}{\ensuremath{\mathbb{Q}}}
\newunicodechar{₊}{\ensuremath{_+}}
\newunicodechar{∗}{\ensuremath{\ast}}
\newunicodechar{₁}{\ensuremath{_1}}
\newunicodechar{₂}{\ensuremath{_2}}
\newunicodechar{⊥}{\ensuremath{\bot}}
\newunicodechar{∘}{\ensuremath{\circ}}
\newunicodechar{ρ}{\ensuremath{\rho}}
\newunicodechar{↦}{\ensuremath{\mapsto}}
\newunicodechar{₀}{\ensuremath{_0}}
\newunicodechar{∙}{\ensuremath{\boldsymbol{\cdot}}}
\newunicodechar{Ω}{\ensuremath{\Omega}}
\newunicodechar{§}{\ensuremath{\mathbb{S}}}
\newunicodechar{𝟙}{\ensuremath{\mathbb{1}}}
\newunicodechar{×}{\ensuremath{\times}}
\newunicodechar{⊔}{\ensuremath{\sqcup}}
\newunicodechar{∞}{\ensuremath{\infty}}
\newunicodechar{＝}{\ensuremath{=}}
\newunicodechar{⁰}{\ensuremath{^0}}
\newunicodechar{↪}{\ensuremath{\hookrightarrow}}

\newcommand{\V}{\mathsf{V}\xspace}
\newcommand{\El}{\mathsf{El}\xspace}
\newcommand{\Vinf}{\ensuremath{\V^\infty}}
\newcommand{\Vz}{\V^0}
\newcommand{\Elzop}{\El^0}
\newcommand{\Elz}[1]{\Elzop\,#1}
\newcommand{\supop}{\mathsf{sup}}
\newcommand{\supc}[2]{\supop\,#1\,#2}
\newcommand{\supz}{\supop^0}
\newcommand{\supzc}[2]{\supz\,\left(#1,#2\right)}
\newcommand{\desupop}{\mathsf{desup}}
\newcommand{\desupz}{\desupop^0}
\newcommand{\ap}{\mathsf{ap}}
\newcommand{\refl}{\mathsf{refl}}

\newcommand{\U}{\mathcal{U}\xspace}
\newcommand{\TThProp}[1]{\mathsf{hProp}_{#1}}
\newcommand{\TThPropU}{\TThProp{\U}}
\newcommand{\TThSet}[1]{\mathsf{hSet}_{#1}}
\newcommand{\TThSetU}{\TThSet{\U}}
\newcommand{\TTempty}{\mathsf{empty}}
\newcommand{\TTemptyelim}{\TTempty\mbox{-}\mathsf{elim}}
\newcommand{\TTunit}{\mathsf{unit}}
\newcommand{\TTtt}{\mathsf{tt}}
\newcommand{\TTbool}{\mathsf{bool}}
\newcommand{\TTtrue}{\mathsf{true}}
\newcommand{\TTfalse}{\mathsf{false}}
\newcommand{\TTifthenelse}[3]{\mathsf{if}\,#1\,\mathsf{then}\,#2\,\mathsf{else}\,#3}
\newcommand{\TTnat}{\mathbb{N}}
\newcommand{\TTnatzero}{\mathsf{0}}
\newcommand{\TTnatsuc}{\mathsf{s}}
\newcommand{\TTFin}[1]{\mathsf{Fin}\,#1}

\newcommand{\TTCoprod}[2]{#1\,+\,#2}
\newcommand{\TTinl}{\mathsf{inl}}
\newcommand{\TTinr}{\mathsf{inr}}
\newcommand{\TTisemb}{\mathsf{is\mbox{-}emb}}
\newcommand{\TTisiterative}{\mathsf{is\mbox{-}iterative\mbox{-}set}}

\newcommand{\TTfst}{\mathsf{pr}_1}
\newcommand{\TTsnd}{\mathsf{pr}_2}
\newcommand{\Wop}{\mathsf{W}}
\newcommand{\W}[2]{\Wop_{#1}\,#2}
\newcommand{\fibop}{\mathsf{fib}}
\newcommand{\fib}[2]{\fibop\,#1\,#2}

\newcommand{\idequiv}{\mathsf{id\mbox{-}equiv}}
\newcommand{\reflhtpy}{\mathsf{refl\mbox{-}htpy}}
\newcommand{\graph}{\mathsf{graph}}
\newcommand{\MEWOcov}{\mathsf{MEWO}_{\mathsf{cov}}}

\newcommand{\Pizop}{\Pi^0}
\newcommand{\Piz}[2]{\Pizop\,#1\,#2}
\newcommand{\Sigmazop}{\Sigma^0}
\newcommand{\Sigmaz}[2]{\Sigmazop\,#1\,#2}
\newcommand{\Idzop}{\mathsf{Id}^0}
\newcommand{\Idz}[3]{\Idzop\,#1\,#2\,#3}
\newcommand{\Coprodz}[2]{#1\,+^0\,#2}
\newcommand{\emptyz}{\TTempty^0}
\newcommand{\unitz}{\TTunit^0}
\newcommand{\boolz}{\TTbool^0}
\newcommand{\natz}{\TTnat^0}
\newcommand{\sucz}{\mathsf{suc}^0}
\newcommand{\Vzcode}[1]{\Vz_{#1}\mbox{-}\mathsf{code}}
\newcommand{\funz}[2]{#1 \to^0 #2}

\newcommand{\prodz}[2]{#1 \times^0 #2}

\newcommand{\fibopz}{\mathsf{fib}^0}
\newcommand{\fibz}[2]{\fibopz\,#1\,#2}

\newcommand{\Ob}{\mathsf{Ob}}
\newcommand{\Hom}{\mathsf{Hom}}
\newcommand{\Vcat}{\mathcal{V}}
\newcommand{\id}{\mathsf{id}}
\newcommand{\hSetcat}[1]{\mathcal{hSet}_{#1}}
\newcommand{\hSetcatU}{\hSetcat{\U}}

\definecolor{dkblue}{rgb}{0,0.1,0.5}
\definecolor{lightblue}{rgb}{0,0.5,0.5}
\definecolor{dkgreen}{rgb}{0,0.6,0}
\definecolor{dkbrown}{rgb}{0.4,0,0}
\definecolor{dkviolet}{rgb}{0.6,0,0.8}

\newcommand{\mycomment}[3]{}
\defcitealias{HoTT13}{HoTT Book}

\begin{document}

\lefttitle{The Category of Iterative Sets in HoTT/UF}

\righttitle{Mathematical Structures in Computer Science}

\papertitle{Article}

\jnlPage{1}{00}
\jnlDoiYr{2019}
\doival{10.1017/xxxxx}

\title{The Category of Iterative Sets in Homotopy Type Theory and Univalent Foundations}

\begin{authgrp}
  \author{Daniel Gratzer}
  \affiliation{Department of Computer Science, Aarhus University, Denmark \email{gratzer@cs.au.dk}}
  \author{Håkon Gylterud}
  \affiliation{Department of Informatics, University of Bergen, Norway \email{hakon.gylterud@uib.no}}
  \author{Anders Mörtberg}
  \affiliation{Department of Mathematics, Stockholm University, Sweden \email{anders.mortberg@math.su.se}}
  \author{Elisabeth Stenholm}
  \affiliation{Department of Informatics, University of Bergen, Norway \email{elisabeth.stenholm@uib.no}}
\end{authgrp}

\history{(Received xx xxx xxx; revised xx xxx xxx; accepted xx xxx xxx)}

\begin{abstract}
  When working in Homotopy Type Theory and Univalent Foundations, the
  traditional role of the category of sets, $\mathcal{Set}$, is
  replaced by the category $\mathcal{hSet}$ of homotopy sets (h-sets);
  types with h-propositional identity types. Many of the properties of
  $\mathcal{Set}$ hold for $\mathcal{hSet}$ ((co)completeness,
  exactness, local cartesian closure, etc.). Notably, however, the
  univalence axiom implies that $\Ob\,\mathcal{hSet}$ is not itself an
  h-set, but an h-groupoid. This is expected in univalent foundations,
  but it is sometimes useful to also have a stricter universe of sets,
  for example when constructing internal models of type
  theory. In this work, we equip the type of iterative sets $\Vz$, due
  to \citet{Gylterud18} as a refinement of the pioneering work of
  \citet{Aczel78} on universes of sets in type theory, with the
  structure of a Tarski universe and show that it satisfies many of
  the good properties of h-sets. In particular, we organize $\Vz$ into
  a (non-univalent strict) category and prove that it is locally
  cartesian closed. This enables us to organize it into a category with
  families with the structure necessary to model extensional type
  theory internally in HoTT/UF. We do this in a rather minimal
  univalent type theory with W-types, in particular we do not rely on
  any HITs, or other complex extensions of type theory.
  Furthermore, the construction of $\Vz$ and the
  model is fully constructive and predicative, while still being very
  convenient to work with as the decoding from $\Vz$ into h-sets
  commutes definitionally for all type constructors. Almost all of the paper
  has been formalized in \Agda using the \agdaunimath library of
  univalent mathematics.
\end{abstract}


\maketitle

\paragraph*{Competing interests:}
The authors declare no competing interests.

\paragraph*{Acknowledgments:}
Anders would like to thank the late Vladimir Voevodsky for interesting
him in the problem of set universes of sets in Univalent
Foundations. In the last email Vladimir sent to Anders (on August 17,
2017) he asked about progress on this problem and mentioned that he
had had a new idea for its solution. Unfortunately Vladimir never
replied with any details about the solution, but we hope that he
would have been pleased with the solution in this paper.

This work is also a testament to influence of the late Peter Aczel,
who already in the late 1970s laid the foundations of this work.
We are also indebted to the late Erik Palmgren, whose research,
teaching and encouragement shaped the constructions herein.

The authors would like to thank Peter LeFanu Lumsdaine for valuable
comments and discussions, and for stating our main result before we
had done so ourselves. We are also very grateful to Henrik Forssell
for helping us get a copy of the relevant chapter of the 1984 master
thesis of \citet{Salvesen84}, which only seems to exist in hard copy
at the University of Oslo library.

Anders Mörtberg has been supported by the Swedish Research Council
(Vetenskapsrådet) under Grant No.~2019-04545.

\section{Introduction}
\label{sec:intro}

Foundational theories of mathematics are concerned with collections of
mathematical objects. Depending on the specific foundation, these collections
might be called sets, classes or types. Among the many schisms of foundational
theories, we find the one between \emph{material} and \emph{structural}. In a
material foundational theory, the objects within a collection have an identity
independent of the collection, and it is a sensible question to compare elements
of different collections by equality. On the other hand, in a structural theory,
the elements of a collection have no identity separate from the collection, and
the important aspects of a collection are how its structure interacts with the
other collections, for instance through functional relations.

Traditional set theories, such as Zermelo--Fraenkel set theory (ZF), are
material foundational theories: there is a global elementhood relation and a
global identity relation, meaning that all objects of the theory are possible
elements of any set and can be compared to any other elements. This gives each
set an inherent structure of membership relations between its elements, the
elements of its elements, and so on. On the other hand, intensional Martin-Löf
type theory (MLTT) \citep{MartinLof75itt} is a structural theory where the
identity type compares only elements of the same type. Furthermore, in
\emph{Homotopy Type Theory and Univalent Foundations}\footnote{We will refer to
the book \emph{Homotopy Type Theory: Univalent Foundations of Mathematics}
\citep{HoTT13} as the ``\citetalias{HoTT13}'' throughout the rest of the paper.}
(HoTT/UF) \citep{HoTT13} the Univalence Axiom (UA) \citep{Voevodsky10cmu} can be
seen in structural terms as saying that \emph{structural equivalence is
identity} \citep{Awodey13}. HoTT/UF also distinguishes its types into
\emph{h-levels}/\emph{$n$-types}: contractible types, h-proposition, h-sets,
h-groupoids, and so on \citep{Voevodsky15unimath}. The h-sets correspond to sets
as realized by other structural set theories, while types of higher h-levels are
(higher-dimensional) groupoids which are not primitive objects in other
foundational theories.

In type theory, the types are organized into universes, and UA is formulated
relative to a specific universe. Thus, one can have both \emph{univalent} and
\emph{non-univalent} universes living side by side. Univalence of a universe is
mostly a positive feature: since every definable operation respects equality,
structures can be transported along equivalences using univalence. One immediate
observation is that in a univalent universe containing at least the booleans,
the subuniverse of all h-sets in that universe cannot itself be an h-set.
However, there are situations where it would be useful to have \emph{a family of
h-sets which itself is an h-set}. One such situation is when constructing the
set model of type theory, as for example a category with families (CwF)
\citep{Dybjer96}, within HoTT/UF. The natural way of doing this would be to
start with a univalent universe $\U$ and attempt to equip the corresponding
category of h-sets ($\hSetcatU$) with a CwF structure. Part of the structure of
a CwF is a presheaf $\mathsf{Ty}$, which is usually formalized in HoTT/UF as a
contravariant functor from the category into h-sets. The objects of the source
category are thought of as contexts and the $\mathsf{Ty}$-functor specifies what
the types are in a given context. The natural choice when organizing $\hSetcatU$
into a CwF would be to let $\mathsf{Ty}(\Gamma) := \Gamma \to \mathsf{hSet}_\U$.
Informally, the types in context $\Gamma$ are simply families of h-sets (in $\U$)
over $\Gamma$\!.  However, since $\mathsf{hSet}_\U$ is not an h-set, this is
ill-typed.

The agenda of this paper is to explore how one specific choice of a cumulative
hierarchy of h-sets, namely the hierarchy $\Vz$ as defined by
\citet{Gylterud18}, can be used as a (non-univalent) universe in HoTT/UF. In
particular, we will study the structural and categorical properties of this
inherently material structure and use it as the basis for a CwF structure.
$\Vz$ is a good starting point for our investigation into internal models of
type theory since its construction uses only elementary type-formers: Π-types,
Σ-types, W-types and identity types. In particular, neither the type $\Vz$
itself nor the ∈-relation defined on it require higher-inductive types,
truncations or quotients. Since $\Vz$ is an h-set, and it is closed under the
usual type formers, it assembles into a model of MLTT with uniqueness of
identity proofs and function extensionality, constructed within MLTT+UA. In this
work, UA plays an essential role. We use it to, for instance, characterize the
identity type of $\Vz$\!. Using UA can sometimes result in constructions which
lack the nice computational properties one has in bare MLTT. In our case
however, since $\Vz$ itself is built from elementary type-formers, many of the
crucial equations, such as the ones for decoding type formers in $\Vz$\!, hold
definitionally. This makes $\Vz$ extremely ergonomic from a formalization
perspective.

Indeed, almost all of this paper has been formalized in the proof assistant \Agda
\citep{Agda}---a dependently typed programming language where one can construct
both programs and proofs using the same syntax. Throughout the paper the \Agda
logo,~\AgdaHen, next to a result is a clickable link to the \Agda code for that
result.  For basic results and constructions in HoTT/UF, we have used the
\agdaunimath library \citep{agda-unimath}---a large \Agda library of formalized
mathematics from the univalent point of view.  Our formalization is in many
places more general than the results presented in this paper as many
constructions used here have a generalization to higher h-levels, and it is
these generalized constructions that have been formalized. They are used for the
Univalent Material Set Theory developed by \citet{GylterudStenholm23}.

\subsection{Formal meta-theory and assumptions}

While our formalization has been carried out in \Agda on top of
\agdaunimath, the results in this paper can be obtained in a more
modest type theory, and are modular in the sense that if you
strengthen the underlying type theory with more types, such as
quotients or more universes, these will be reflected in the internal
model. The majority of the results assume that one works in MLTT
extended with UA. By ``MLTT'' we take an intensional version of MLTT
with the same types and type formers as in \citep[Table 2]{HAN},
namely:

\begin{itemize}
 \item Π-types, denoted $∏_{x:A}B(x)$ with application denoted by juxtaposition and λ-abstraction by $λ(x:A).b(x)$.
 \item Σ-types, denoted $Σ_{x:A}B(x)$ with projections $\TTfst$ and $\TTsnd$.
 \item W-types, denoted $\W{x:A}B(x)$ with canonical elements $\supc{A}{f}$.
 \item Identity types, denoted $a=a'$\!, sometimes subscripted $a=_Aa'$ for clarity, with reflective elements $\refl_a : a = a$.
 \item Binary sum types, denoted $A+B$ with injections $\TTinl$ and $\TTinr$.
 \item Base types: $\TTempty$, $\TTunit$, $\TTbool$, ℕ, with $\TTtt$ being the canonical
   element of $\TTunit$, $\TTtrue, \TTfalse$ the elements of $\TTbool$, and elements of ℕ denoted by $\TTnatzero$ and $\TTnatsuc\,n$. We also let $\TTFin{n}$ denote the type with $n$ elements.
 \item Universes, denoted $\U$, closed under the aforementioned type formers.
       For constructions needing more than one universe level,
       we will subscript them $\U_0, \U_1, ⋯, \U_\ell, ⋯$.
\end{itemize}

One important difference to \citep{HAN} though is that we of course do
not assume equality reflection and instead have intensional identity
types as in \citep{MartinLof75itt}. Another difference is that we, for
convenience, assume definitional $\eta$ for $\Sigma$-types. Our system
is hence very similar to Martín Escardó's spartan MLTT
\citep{Escardo19} and the basic system used in \UniMath
\citep{UniMath}, but with the addition of W-types. The only
construction going beyond this is the construction of set quotients in
Section \ref{sec:vsetquotients}, which assumes that the universe has
set quotients.  The construction of subuniverses in Section
\ref{sec:vuniverses} also naturally assumes that the starting universe
has subuniverses as well. But even with these extensions, the
development is completely constructive and predicative, in particular
we do not rely on LEM, AC, or any resizing principles~\citep{Voevodskyresizing}.

For convenience, we also rely on definitions and notational conventions from the
\citetalias{HoTT13}. Among these are:

\begin{itemize}
 \item Definitional/judgmental equality is denoted by ${≡}$.
 \item Homotopy of functions is denoted $f ∼ g$, with $\reflhtpy$ denoting $λa.\refl$
 \item Type equivalence is denoted $A ≃ B$ with identity equivalence
       $\idequiv : A ≃ A$.
 \item h-levels/$n$-types, in this paper we mainly work with types in $\TThPropU$ and $\TThSetU$, i.e.\ the h-propositions and
   h-sets in a given universe.
 \item We use pattern-matching freely in definitions and proofs,
   instead of explicit eliminators.
\end{itemize}

\subsection{Contributions of the paper}

While Section \ref{sec:vdef} sets the stage by recounting the definition of
$\Vz$\!, the foundation of this paper's contributions is built in Section
\ref{sec:vtarski}: we show that $\Vz$ forms a Tarski-style universe closed under
Π-types, Σ-types, identity types, coproducts, set quotients, and that it contains
basic types like $\TTempty$, $\TTunit$, $\TTbool$, ℕ and a hierarchy of subuniverses.
Proposition \ref{prop:equiv-emb-int} is central to this, as it characterizes the
small types representable in $\Vz$ as those which can be embedded into it. Since
the decoding of all type formers is definitional, this gives an ergonomic universe
of h-sets which itself is an h-set, which can be used in HoTT/UF.
In Section \ref{sec:vcategory} we shed light on the categorical properties of
$\Vz$\!. In particular, we show that it is a locally cartesian closed category,
with finite limits and colimits, and that there is a full and faithful functor
back to $\hSetcatU$ which preserves this structure.
The final technical contribution is the construction of an extensional model of
MLTT internal in MLTT+UA, based on $\Vz$\!. This is done by giving CwF structure
to $\Vz$\!. The formalization of this includes contributions to \agdaunimath,
in particular the definition of a CwF with associated structure.
A bibliographic contribution can be found in Section \ref{sec:otherv},
where we compare our constructions to existing developments on the
relationship between set theory and type theory. This relationship
has taken many forms over the years, and goes back to the 1970s.

\section{Definition of \texorpdfstring{$\Vz$}{V0} and its basic properties}
\label{sec:vdef}

The ideas behind $\Vz$ trace back to \emph{The type theoretic interpretation of constructive set
  theory} by \citet{Aczel78}. In \emph{op.\ cit.}, \citeauthor{Aczel78} constructed a model of set
theory in dependent type theory relying upon a non-trivial defined equality relation on the
underlying type of the model in order to (hereditarily) force set-extensionality. This underlying
type of the model is what we in modern parlance would call a $\Wop$-type.
To construct $\Vz$\!, we opt to carve out a subtype of a $\Wop$-type rather than take such a
quotient. Instead of defining an equivalence relation which identifies the elements of the
$\Wop$-type which represent the same set, we shall identify a subtype of the $\Wop$-type which
contains only the canonical representations of each (iterative) set. Thus, we get a model of set
theory in type theory where the equality \emph{is} interpreted as the identity type and no further
non-trivial  identifications are required.

In this section we will review the definition of $\Vz$ and prove some
properties about it. In particular, we will show that $\Vz$ is an
h-set. In order to define $\Vz$\!, we start by recalling the $\Wop$-type Aczel
used: ``the unrestricted iterative hierarchy''. It is the type of well-founded trees
with branching types chosen freely from a fixed universe $\U_\ell$.

\begin{definition}[\agdalink{https://elisabeth.stenholm.one/category-of-iterative-sets/iterative.set.html\#2146}]
  Given a universe $\U_\ell$, we define the type $\Vinf_\ell$ as
  \[
    \Vinf_\ell := \W{A : \U_\ell}{A}
  \]
\end{definition}

We will usually omit the universe level $\ell$ for $\U_\ell$ and $\Vinf_\ell$, and
write simply $\U$ and $\Vinf$\!.
When seeing $\Vinf$ as a type of sets, an element
$\supc{A}{f} : \Vinf$ represents a set whose elements are indexed by
the type $A : \U$. The function $f : A → \Vinf$ picks out the
element at each index.
Since the function $f$ need not be injective,
the same element can be picked out several times. Indeed, the role of Aczel's equivalence
relation on this type was to erase this multiplicity. If we instead omit this further
identification $\Vinf$ can be seen as a type of multisets~\citep{Gylterud20}.

\begin{notation}
  Given $x : \Vinf$, we follow \citet{Aczel78} and define a pair of operations
  $\overline{x} : \U$ and $\widetilde x : \overline x → \Vinf$, as follows:
  \begin{mathpar}
    \overline{\supc{A}{f}} \defeq A
    \and
    \reallywidetilde{\supc{A}{f}} \defeq f
  \end{mathpar}
\end{notation}

We present two characterizations of equality in $\Vinf$ as both are
useful in different contexts. We note that both characterizations rely on univalence.

The first is an instance of a more general characterization of equality in $\Wop$-types
\citep[Lemma~1]{Gylterud20}. It states that two elements are equal if they have equivalent
underlying indexing types and this equivalence is coherent with respect to the functions picking out
the elements.

\begin{therm}[{{\citep[Theorem 1]{Gylterud20}}}, \agdalink{https://elisabeth.stenholm.one/category-of-iterative-sets/iterative.set.html\#2806}]\label{thm:Vinf-eq}
  For two elements $x,y : \Vinf$ the canonical map
  \[
    \left(x = y\right) → \left(∑_{e :\,\overline{x} ≃ \overline{y}} \widetilde{x} ∼ \widetilde{y} ∘ e\right)
  \]
  which sends $\refl$ to $(\idequiv, \reflhtpy)$, is an equivalence.
\end{therm}

The second characterization of equality in $\Vinf$ states that two
elements in $\Vinf$ are equal when the functions picking out the
elements are fiberwise equivalent. Intuitively, this means that they
pick out the same elements the same number of times. One can think of
this characterization of equality as a higher level generalization of
the axiom of extensionality.

\begin{therm}[{{\citep[Theorem 2]{Gylterud20}}}, \agdalink{https://elisabeth.stenholm.one/category-of-iterative-sets/e-structure.from-T-coalgebra.html\#1124}]\label{thm:ext}
  For two elements $x,y : \Vinf$ the canonical map
  \[
    \left(x = y\right) → ∏_{z : \Vinf} \fib{\widetilde{x}}{z} ≃ \fib{\widetilde{y}}{z}
  \]
  which sends $\refl$ to $λz.\idequiv$, is an equivalence.
\end{therm}
\begin{proof}
  We reproduce the proof here for convenience. We have the following chain of
  equivalences:
  \begin{align*}
    \left(x = y\right)
    ≃ \left(∑_{e :\,\overline{x}\,≃\,\overline{y}} \widetilde{x} ∼ \widetilde{y} ∘ e\right)
    ≃ \left(∏_{z : \Vinf} \fib{\widetilde{x}}{z} ≃ \fib{\widetilde{y}}{z}\right)
  \end{align*}

  The first equivalence is the one constructed in Theorem \ref{thm:Vinf-eq}. The second equivalence
  is proven by \citet[Lemma 5]{Gylterud20}. One directly checks that the constructed equivalence
  computes as desired for $\refl$.
\end{proof}

We will not dwell much on our structures being models of material set theory,
but rather focus on their structural properties in this paper.  However, we will
define the elementhood relation on $\Vinf$ following \citet{Gylterud20}. This
elementhood relation, and its well-foundedness, will be used in later
constructions.

\begin{definition}[Elementhood, \agdalink{https://elisabeth.stenholm.one/category-of-iterative-sets/e-structure.from-T-coalgebra.html\#1047}]
  We define $∈\,:\,\Vinf → \Vinf → \U$ by
  \[
    x ∈ y := \fib{\widetilde{y}}{x}
  \]

  In particular, for canonical elements we get
  \[
    \left(x ∈ \supc{A}{f} \right) ≡ \left(∑_{a:A}f\,a=x\right)
  \]
\end{definition}

The exensionality property of Theorem \ref{thm:ext} can now be reformulated
as an equivalence
\[
  \left(x = y\right) ≃ \left(∏_{z:\Vinf} z ∈ x ≃ z ∈ y\right)
\]
By virtue of univalence, we obtain this extensionality result without taking
quotients by set extensionality or bisimulation like Aczel does. In particular,
we are able to avoid working with quotients or setoids while still achieving the
equivalence above.

Note that $x ∈ y$ need not be an h-proposition, i.e., $y$ could contain several instances of
$x$. This is because, as discussed above, there is no restriction on the function $\widetilde{y}$
and its fibers. We will soon focus our attention to a subtype of $\Vinf$ where these fibers are
h-propositions, i.e., where they have at most one inhabitant. But first, we will look at how some
familiar sets can be represented in $\Vinf$\!.

We define the empty set as follows:
\[
  ∅ := \supc{\TTempty}{\TTemptyelim}
\]
This represents the empty set since for any $x : \Vinf$, the type $x ∈ ∅$
is empty.

Given $x : \Vinf$ we can construct the singleton containing $x$ as follows:
\[
  \{ x \} := \supc{\TTunit}{(λ\_.x)}
\]
The type $x ∈ \{ x \}$ is inhabited by $(\TTtt,\refl)$. Indeed, for any $y :
\Vinf$, there is an equivalence $\left(y ∈ \{ x \} \right) ≃ \left(y = x\right)$.

We can also construct the unordered pair of two elements $x, y : \Vinf$:
\[
  \{ x, y \} := \supc{\TTbool}{(λb.\TTifthenelse{b}{x}{y})}
\]
For any $z : \Vinf$, the type $z ∈ \{x,y\}$ is equivalent to $\TTCoprod{(z =
x)}{(z = y)}$. Note in particular that the type $x ∈ \{x,x\}$ is equivalent to
$\TTCoprod{(x = x)}{(x = x)}$, which contains at least two distinct elements. Thus,
$\{x,x\}$ is a multiset which contains two copies of $x$. Using images
one can whittle this down to an iterative set, see the forthcoming paper
\citep{GylterudStenholm23} for details on the various types of pairing in
higher h-levels.

In order to construct a universe of sets we need to ensure that the
$∈$-relation is h-proposition valued, i.e., that any element occurs at most
once in a set. As the type $x ∈ y$ is the type of homotopy fibers of
$\widetilde{y}$ over $x$, this type would be an h-proposition if $\widetilde{y}$
was an embedding:

\begin{definition}[{{\citepalias[Definition 4.6.1]{HoTT13}}}, \agdalink{https://elisabeth.stenholm.one/category-of-iterative-sets/foundation-core.embeddings.html\#1086}]
  A function $f : A → B$ is an \emph{embedding} if
  $\ap\,f\,x\,y : x = y → f\,x = f\,y$ is an equivalence for all
  $x\,y : A$.
\end{definition}

We write $\TTisemb\,f$ for the type of proofs that $f$ is an embedding and $f :
A ↪ B$ for ${∑}_{f : A → B}\,\TTisemb\,f$. A key
observation about embeddings is:

\begin{lemma}[{{\citepalias[Lemma 7.6.2]{HoTT13}}}, \agdalink{https://elisabeth.stenholm.one/category-of-iterative-sets/foundation-core.propositional-maps.html\#2019}]\label{lem:embedding-propfiber}
  A function $f : A → B$ is an embedding if and only if it has h-propositional fibers.
\end{lemma}

This motivates Gylterud's definition of iterative sets in HoTT/UF \citep{Gylterud18}:

\begin{definition}[Iterative sets, \agdalink{https://elisabeth.stenholm.one/category-of-iterative-sets/iterative.set.html\#3053}]
  We define $\TTisiterative : \Vinf → \U$ as
  \[
    \TTisiterative\,(\supc{A}{f}) := (\TTisemb\,f) \times \left( ∏_{a : A} \TTisiterative\,(f\,a) \right)
  \]
\end{definition}

The idea is to pick out those elements $x : \Vinf$ for which the function that selects elements is
an embedding and such that the elements of $x$ satisfy the same criterion, recursively. This means
that any $y : \Vinf$ element is a member of $x$ at most once and, consequently, $x$ encode a set
rather than a multiset. For these sets the $∈$-relation becomes h-proposition valued by
Lemma~\ref{lem:embedding-propfiber}, as desired.

Not all the elements in $\Vinf$ are iterative sets. For example, the
unordered pair $\{ ∅, ∅ \}$ from above is \emph{not} an iterative set
as the function in the definition is not an embedding.\footnote{There is
  a different way to construct pairs which does yield an iterative set when
  applied to iterative sets. For
  details, see the proof of the the axioms of Myhill's constructive set
  theory given by \citet{Gylterud18}.}
On the other hand, the empty set, $∅$, is an iterative set, since $\TTemptyelim$
is always an embedding, regardless of the codomain. Moreover, for any iterative
set $x : \Vz$, the singleton $\{x\}$ is an iterative set since any map from an
h-proposition into an h-set is an embedding (we will see that $\Vz$ is an h-set
in Theorem~\ref{thm:Vzhset}). Furthermore, if $x$ and $y$ are distinct iterative
sets then $\{ x, y \}$ is also an iterative set. To see this, it suffices to
verify that the below map $φ : \TTbool → \Vz$ is an embedding if it is injective:
\[
  φ \, b := \TTifthenelse{b}{x}{y}
\]
Given $b_1, b_2 : \TTbool$,
either $b_1 = b_2$, in which case we are done, or $b_1 \neq b_2$,
in which case we get a path between $x$ and $y$,
from which the result follows by assumption.

\begin{definition}[Type of iterative sets, \agdalink{https://elisabeth.stenholm.one/category-of-iterative-sets/iterative.set.html\#3823}]
  \label{def:v}
  We define the type of \emph{iterative sets} as follows:
  \[
    \Vz := {∑}_{x : \Vinf}\,\TTisiterative\,x
  \]
\end{definition}

We will extend the previously introduced notation to apply to iterative sets:
\[
  \overline{(\supc{A}{f}, p)} := A \qquad \reallywidetilde{(\supc{A}{f}, p)} := f
\]
Moreover, the elementhood relation $∈$ defined on $\Vinf$ gives rise to an elementhood relation for
$\Vz$ given by projecting out the underlying elements in $\Vinf$ and applying $∈$: for any
$x,y : \Vz$ we let $x ∈ y := \TTfst \, x ∈ \TTfst \, y$.  We use the same notation for both
relations, as it will be clear from context which one is meant.

\begin{lemma}[\agdalink{https://elisabeth.stenholm.one/category-of-iterative-sets/iterative.set.html\#3205}]
  For all $x : \Vinf$, the type $\TTisiterative\,x$ is an h-proposition.
\end{lemma}
\begin{proof}
  This follows by induction on $x : \Vinf$, together with the fact that being an
  embedding is an h-proposition.
\end{proof}

\begin{corollary}[\agdalink{https://elisabeth.stenholm.one/category-of-iterative-sets/iterative.set.html\#4044}]
  The projection $\TTfst : \Vz → \Vinf$ is an embedding, i.e.\ $\Vz$ is a subtype of $\Vinf$\!.
\end{corollary}
\begin{proof}
  This is an instance of the fact that for any type $A$ and family $P$ of
  h-propositions over $A$, the first projection $\TTfst : ∑_{a : A} P\,a → A$
  is an embedding.
\end{proof}

Having an embedding $\Vz↪\Vinf$ means that equality in $\Vz$ is exactly equality of
the corresponding elements in $\Vinf$\!. Since we have already characterized equality
in $\Vinf$, we can use this characterization to show that $\Vz$ is an h-set.

\begin{therm}[\agdalink{https://elisabeth.stenholm.one/category-of-iterative-sets/iterative.set.html\#17196}]\label{thm:Vzhset}
  $\Vz$ is an h-set.
\end{therm}
\begin{proof}
  For $(x,p), (y,q) : \Vz$ we have a chain of equivalences:
  \[
    \left((x,p) =_{\Vz} (y,q)\right)
    ≃ \left(x =_{\Vinf} y\right)
    ≃ \left(∏_{z : \Vinf} z ∈ x ≃ z ∈ y\right)
  \]
  The first equivalence is the characterization of equality in subtypes. The
  second is Theorem \ref{thm:ext}. Note that $z ∈ x \equiv
  \fib{\widetilde{x}}{z}$, and $\widetilde{x}$ is an embedding by $p$. Thus $z
  ∈ x$ is an h-proposition. The same holds for $z ∈ y$. Thus, the rightmost
  type in the chain of equivalences above is a family of equivalences between
  h-propositions, and is thus an h-proposition. It then follows that the type
  $(x,p) =_{\Vz} (y,q)$ is an h-proposition.
\end{proof}

Given a type $A : \U$ and an embedding $f : A ↪ \Vz$, we can
construct an element of $\Vz$\!. This function is the counterpart to $\supop$ for
$\Vinf$, and while it is not formally a constructor it behaves like one in that
the recursion and elimination principles, with fitting computation rules,
are provable for it \citep{GylterudStenholm23}.

\begin{remark}
The underscores in the constructions below denote proof terms for the h-propositions involved. We
omit these for readability, and refer the interested reader to the formalization.
\end{remark}

\begin{definition}[\agdalink{https://elisabeth.stenholm.one/category-of-iterative-sets/iterative.set.html\#5749}]
  We define the following function:
  \begin{align*}
    &\supz : \left(∑_{A : \U} A ↪ \Vz\right) → \Vz \\
    &\supz\,(A,f) := \left(\supc{A}{(\pi_0 ∘ f)}, \_\right)
  \end{align*}

\end{definition}

Similarly, given an element of $\Vz$, we can extract the underlying type and
embedding.

\begin{definition}[\agdalink{https://elisabeth.stenholm.one/category-of-iterative-sets/iterative.set.html\#6288}]
  We define the following function:
  \begin{align*}
    &\desupz : \Vz → \left(∑_{A : \U} A ↪ \Vz\right) \\
    &\desupz\,(\supc{A}{f}, \_) := (A, (f, \_))
  \end{align*}
\end{definition}

By virtue of being a $\Wop$-type, $\Vinf$ is the initial algebra to
the polynomial functor
\[
  X ↦ \left(∑_{A : \U} A → X\right)
\]

Similarly, $\Vz$ is the initial algebra for the functor $X ↦ \left(∑_{A : \U} A
↪ X\right)$, even though this functor is not polynomial. The initiality induces
an equivalence $\Vz ≃ \left(∑_{A : \U} A ↪ \Vz\right)$, realized by the maps
$\supz$ and $\desupz$ above. These results are due to
\citet{GylterudStenholm23}, who extend this construction to a whole hierarchy of
functors $X ↦ \left(∑_{A : \U} A ↪_n X\right)$, for $n : \TTnat_{-1}$. Each of
these have an initial algebra, given by a higher level generalization of $\Vz$\!.

\section{\texorpdfstring{$\Vz$}{V0} as a universe à la Tarski}
\label{sec:vtarski}

The type $\Vz$ can be thought of as a type of material sets, in the sense that
$\Vz$ together with the binary relation $\in$ is a model of
constructive set theory \citep{Gylterud18}. This section demonstrates
that, more type-theoretically, $\Vz$ can be organized into a
universe à la Tarski. In this way, $\Vz$ becomes a universe of h-sets which is itself
an h-set. Furthermore, $\Vz$ is a strict universe in the sense that the decoding
from codes to types is definitional. For instance, the decoding of the code for the natural numbers
is definitionally equal to the type of natural numbers, and the decoding
of a Π- or Σ-type of a family is the actual Π- or Σ-type of the decoding
of the family.

We begin by defining the decoding family for our universe, $\Vz$\!, as the
underlying index type for each of its elements.

\begin{definition}[Decoding, \agdalink{https://elisabeth.stenholm.one/category-of-iterative-sets/iterative.set.html\#6525}]
  We define the decoding function $\Elz : \Vz \to \U$ by
  \[
    \Elz x :=\overline x
  \]
\end{definition}

It is easy to prove that the decoding of each code in $\Vz$ is also an h-set:

\begin{therm}[\agdalink{https://elisabeth.stenholm.one/category-of-iterative-sets/iterative.set.html\#17117}]
  For every $x : \Vz$ the type $\Elz{x}$ is an h-set.
\end{therm}
\begin{proof}
  Recall that $\widetilde x$ embeds $\Elz x$ into $\Vz$\!.
  By Theorem \ref{thm:Vzhset} $\Vz$ is an h-set. Since any type which embeds
  into an h-set is an h-set, it follows that $\Elz x$ is an h-set.
\end{proof}

Note that for any $A : \U$ and embedding $f : A ↪ \Vz$ we have the definitional
equality $\Elz{(\supzc A f)} ≡ A$. That is, if we construct a code for a type in
$\U$ using $\supz$ (which is what we usually do), then the decoding of this
code is definitionally equal to the type we started with. This is very
convenient when working with the universe $\Vz$\!, especially for formalization.

As a universe, $\Vz$ contains codes of all the traditional type formers
as long as they are present in the underlying universe, $\U$. Using $\supz$\!, one
can construct a code for a given type $A : \U$ in $\Vz$ if there is an embedding
$A ↪ \Vz$\!. In fact, there is a code for $A$ in $\Vz$ precisely
when it can be embedded into $\Vz$\!.

\begin{proposition}[\agdalink{https://elisabeth.stenholm.one/category-of-iterative-sets/fixed-point.internalisations.html\#1344}]\label{prop:equiv-emb-int}
  For any $A : \U$ there is an equivalence
  \[
    \left(A ↪ \Vz\right) ≃ \left(\sum_{a : \Vz} \Elz{a} = A\right)
  \]
\end{proposition}
\begin{proof}
  The maps back and forth are
  \begin{align*}
    &α : \left(A ↪ \Vz\right) → \sum_{a : \Vz} \Elz{a} = A \\
    &α\,f := (\supzc{A}{f}, \refl) \\
    \vspace{10pt} \\
    &β : \left(\sum_{a : \Vz} \Elz{a} = A\right) → (A ↪ \Vz) \\
    &β\,(a , \refl) := \widetilde{a}
  \end{align*}

  We compute as follows:
  \begin{gather*}
    α(β(a,\refl)) ≡ α (\widetilde a) ≡ (\supzc{\overline a}{\widetilde a},\refl) = (a,\refl)
    \\
    \TTfst \left(β(α␣f)\right) ≡ \TTfst \left(β(\supzc{A}{f}, \refl)\right) ≡ \TTfst \left(\reallywidetilde{\supzc{A}{f}}\right) ≡ \TTfst\,f
  \end{gather*}

  We emphasize that the definitional equation $\Elz{(\supzc A f)} ≡ A$ simplifies the definition of
  α as we may then use $\refl$ for the second argument. Moreover, $β ∘ α$ definitionally preserves
  the function underlying the embedding. The same is not true of the \emph{witness} that this
  function is an embedding, but such witnesses belong to a contractible type and can safely be ignored.
\end{proof}

\subsection{Basic types}

We now construct codes for some basic types in $\Vz$\!.

\begin{proposition}[\agdalink{https://elisabeth.stenholm.one/category-of-iterative-sets/iterative.set.html\#17436}]\label{prop:empty-unit-booleans}
  $\Vz$ contains the empty type, unit type and booleans.
\end{proposition}
\begin{proof}
  We define the elements $\emptyz, \unitz, \boolz : \Vz$ as follows:
  \begin{align*}
    \emptyz &:= ∅
    \\
    \unitz &:= \{∅\}
    \\
    \boolz &:= \{ ∅, \{∅\} \}
  \end{align*}
  There were all verified to be iterative sets in Section \ref{sec:vdef}.
\end{proof}

We note that the expected equations hold up to definitional equality:
\begin{mathpar}
  \Elz{\emptyz} ≡ \TTempty,
  \qquad
  \Elz{\unitz} ≡ \TTunit
  \qquad \text{and} \qquad
  \Elz{\boolz} ≡ \TTbool.
\end{mathpar}

\begin{proposition}[\agdalink{https://elisabeth.stenholm.one/category-of-iterative-sets/iterative.set.html\#20911}]
  $\Vz$ is closed under the natural numbers.
\end{proposition}
\begin{proof}
  By Proposition \ref{prop:equiv-emb-int} it is enough to construct an embedding
  $\TTnat ↪ \Vz$\!. Here there is a choice of encoding of the naturals in $\Vz$
  and several encodings are possible. We will use the von Neumann encoding and
  show that this is an embedding.

  First we define the successor function in $\Vz$:
  \begin{align*}
    &\sucz : \Vz → \Vz \\
    &\sucz\,x := \supz\,(\TTCoprod{\overline{x}}{\TTunit}, \varphi_{x})
  \end{align*}
  In the above, $\varphi_{x} : \TTCoprod{\overline{x}}{\TTunit} → \Vz$ is defined as follows:
  \begin{align*}
    &\varphi_{x}\,(\TTinl\,a) := \widetilde{x}\, a \\
    &\varphi_{x}\,(\TTinr\,b) := x
  \end{align*}
  To see that the map $\varphi_{x}$ is an embedding, note that for any $z : \Vz$
  the fiber $\fib{\varphi_x}{z}$ is equivalent to $\TTCoprod{(z \in x)}{(x =
  z)}$.  Both summands are h-propositions and they are disjoint: if they were
  both inhabited we could derive $x \in x$ which contradicts the
  well-foundedness of $\in$~\citep{GylterudStenholm23}.

  The von Neumann encoding of the natural numbers is then the function:
  \begin{align*}
    &f : \TTnat → \Vz \\
    &f\,\TTnatzero := ∅ \\
    &f\,(\TTnatsuc\,n) := \sucz\,(f\,n)
  \end{align*}

  It remains to show that $f$ is an embedding. As $\TTnat$ and $\Vz$ are both
  h-sets it suffices that $f$ is injective. Observe that
  $\overline {f\,x} ≃ \TTFin{x}$, so if $f\,n = f\,m$ then
  $\TTFin{n} ≃ \TTFin{m}$ from which $n = m$ follows.

  Having shown that $f : \TTnat → \Vz$ is an embedding, we define the (code for
  the) natural numbers in $\Vz$ as follows:
  \[
    \natz := \supzc \TTnat f \qedhere
  \]
\end{proof}

Note, again, that the decoding holds up to definitional equality:
\[
  \Elz{\natz} ≡ \TTnat
\]

\subsection{Type formers}

We now turn to closing $\Vz$ under the standard type formers. For these constructions we will need
ordered pairs.

\begin{lemma}[\agdalink{https://elisabeth.stenholm.one/category-of-iterative-sets/fixed-point.unordered-tupling.html\#10439}]
  There is an ordered pairing operation $〈\_,\_〉 : \Vz × \Vz ↪
  \Vz$\!.
\end{lemma}
\begin{proof}
  Ordered pairs are constructed using the Norbert Wiener encoding. The details
  of this construction can be found in the proof of \citep[Theorem~7]{GylterudStenholm23}.
\end{proof}

\begin{proposition}[\agdalink{https://elisabeth.stenholm.one/category-of-iterative-sets/iterative.set.html\#18781}]
  $\Vz$ is closed under $Π$-types.
\end{proposition}
\begin{proof}
  Let $x : \Vz$ and $y : \Elz{x} → \Vz$\!. By
  \citep[Lemma 12]{GylterudStenholm23} there is an embedding:
  \[
    \graph_{x,y} : \left(∏_{a : \Elz{x}} \Elz{(y\,a)}\right) ↪ \Vz
  \]
  This map sends $\varphi : ∏_{a : \Elz{x}} \Elz{(y\,a)}$ to the element
  $\supz\,\left(\Elz{x},λa.〈\widetilde{x}\,a,\widetilde{(y\,a)}\,(\varphi\,a)〉\right)$.
  The $Π$-type is then defined as follows:
  \[
    \Piz{x}{y} := \supz\,\left(∏_{a : \Elz{x}}  \Elz{(y\,a)}, \graph_{x,y}\right) \qedhere
  \]
\end{proof}

The decoding holds up to definitional equality:
\[
  \Elz{\left(\Piz{x}{y}\right)} ≡ ∏_{a : \Elz{x}} \Elz{(y\,a)}
\]

\begin{corollary}[\agdalink{https://elisabeth.stenholm.one/category-of-iterative-sets/iterative.set.html\#19551}]
  $\Vz$ is closed under (non-dependent) function types. Let $\funz{x}{y}$ denote
  the code for the type $\Elz{x} → \Elz{y}$.
\end{corollary}

\begin{proposition}[\agdalink{https://elisabeth.stenholm.one/category-of-iterative-sets/iterative.set.html\#19098}]
  $\Vz$ is closed under $Σ$-types.
\end{proposition}
\begin{proof}
  Let $x : \Vz$ and $y : \Elz{x} → \Vz$\!. Define a putative embedding as follows:
  \begin{align*}
    &f : \left(∑_{a : \Elz{x}} \Elz{(y\,a)}\right) → \Vz \\
    &f\,(a,b) := 〈\widetilde{x}\,a,\widetilde{(y\,a)}\,b〉
  \end{align*}
  This is the composition of two embeddings: $〈\_,\_〉$ and
  $λ(a,b).(\widetilde{x}\,a,\widetilde{(y\,a)}\,b)$ and therefore an embedding. The last function is
  an embedding because $\widetilde{x}$ is an embedding and as is $\widetilde{(y\,a)}$ for every
  $a : \Elz{x}$. We may now define the code for $Σ$-types:
  \[
    \Sigmaz{x}{y} := \supz\,\left(∑_{a : \Elz{x}} \Elz{(y\,a)}\,f\right) \qedhere
  \]
\end{proof}

The decoding holds up to definitional equality:
\[
  \Elz{\left(\Sigmaz{x}{y}\right)} ≡ ∑_{a : \Elz{x}} \Elz{(y\,a)}
\]

\begin{corollary}[\agdalink{https://elisabeth.stenholm.one/category-of-iterative-sets/iterative.set.html\#19606}]
  $\Vz$ is closed under cartesian products. Let $\prodz{x}{y}$ be
  the code for $\Elz{x} × \Elz{y}$.
\end{corollary}

In order to construct coproducts in $\Vz$ we need two lemmas about embeddings.

\begin{lemma}[\agdalink{https://elisabeth.stenholm.one/category-of-iterative-sets/foundation.propositional-maps.html\#2885}]\label{lma:emb-fix-pr1}
  Given types $Y$, $Z$ and h-set $X$ with a point $x_0 : X$, any embedding $f : X
  × Y ↪ Z$ gives rise to an embedding by fixing the first
  coordinate: $f\,(x_0,\_) : Y ↪ Z$.
\end{lemma}
\begin{proof}
  We need to show that for any $z : Z$, the fiber of $f\,(x_0,\_)$ over $z$ is
  an h-proposition. But the following chain of equivalences holds:
  \[
    \left(∑_{y : Y} f\,(x_0,y) = z\right)
      ≃ \left(∑_{y : Y} ∑_{∑_{x : X} (x = x_0)} f\,(x,y) = z\right)
      ≃ \left(∑_{((x,y),p) : \fib{f}{z}} x = x_0\right)
  \]
  The last type is an h-proposition since $\fib{f}{z}$ is an h-proposition by Lemma~\ref{lem:embedding-propfiber} and
  for each $((x,y),p) : \fib{f}{z}$, the type $x = x_0$ is an h-proposition.
\end{proof}

\begin{lemma}[\agdalink{https://elisabeth.stenholm.one/category-of-iterative-sets/foundation.equality-coproduct-types.html\#8672}]\label{lma:emb-coprod}
  Given types $X$, $Y$, and $Z$ together with embeddings $f : X ↪
  Z$ and $g : Y ↪ Z$. If $f\,x~\neq~g\,y$ for all $x : X$ and $y :
  Y$ then the following map is an embedding:
  \begin{align*}
    &h : \TTCoprod{X}{Y} → Z \\
    &h\,(\TTinl\,x) := f\,x \\
    &h\,(\TTinr\,y) := g\,y
  \end{align*}
\end{lemma}
\begin{proof}
  Let $s,t : \TTCoprod{X}{Y}$. We need to show that $\ap\,h : s = t → h\,s =
  h\,t$ is an equivalence. Using induction on coproducts, there are two kinds of
  cases to consider: when $s$ and $t$ lie in different summands, and when they
  lie in the same one.

  First, suppose without loss of generality that $s ≡ \TTinl\,x$ and $t ≡
  \TTinr\,y$. In this case we need to show that $\ap\,h : \TTinl\,x = \TTinr\,y
  → f\,x = g\,y$ is an equivalence. But both types are empty, so any map between
  them is an equivalence.

  Now, suppose without loss of generality that $s ≡ \TTinl\,x$ and $t ≡
  \TTinl\,x'$. We need to show that $\ap\,h : \TTinl\,x = \TTinl\,x' → f\,x =
  f\,x'$ is an equivalence. But note that the following diagram commutes:
  \begin{center}
    \begin{tikzcd}
      x = x' \arrow[rr, "{\ap\,\TTinl}"] \arrow[rd, "{\ap\,f}"'] &                & {\TTinl\,x = \TTinl\,x'} \arrow[ld, "{\ap\,h}"] \\
                                                                  & {f\,x = f\,x'} &
    \end{tikzcd}
  \end{center}
  Since both $\ap\,f$ and $\ap\,\TTinl$ are equivalences it follows that
  $\ap\,h$ is an equivalence.
\end{proof}

\begin{proposition}[\agdalink{https://elisabeth.stenholm.one/category-of-iterative-sets/iterative.set.html\#19661}]
  $\Vz$ is closed under coproducts.
\end{proposition}
\begin{proof}
  Let $x,y : \Vz$\!. Define the map
  \begin{align*}
    &f : \TTCoprod{\Elz{x}}{\Elz{y}} → \Vz \\
    &f\,(\TTinl\,a) := 〈\emptyz, \widetilde{x}\,a〉 \\
    &f\,(\TTinr\,b) := 〈\unitz, \widetilde{y}\,b〉
  \end{align*}
  By Lemma \ref{lma:emb-fix-pr1} both $λa.〈\emptyz, \widetilde{x}\,a〉$
  and $λb.〈\unitz, \widetilde{y}\,b〉$ are embeddings. Moreover,
  suppose $〈\emptyz, \widetilde{x}\,a〉 = 〈\unitz, \widetilde{y}\,b〉$ for some $a
  : \Elz{x}$ and $b : \Elz{y}$. It then follows that $\emptyz = \unitz$, which
  is a contradiction. Therefore, by Lemma \ref{lma:emb-coprod} we conclude that
  $f$ is an embedding.

  We now define the coproduct:
  \[
    \Coprodz{x}{y} := \supz\,(\TTCoprod{\Elz{x}}{\Elz{y}}, f) \qedhere
  \]
\end{proof}

Note that the decoding holds up to definitional equality:
\[
  \Elz{(\Coprodz{x}{y})} ≡ \TTCoprod{\Elz{x}}{\Elz{y}}
\]

\begin{proposition}[\agdalink{https://elisabeth.stenholm.one/category-of-iterative-sets/iterative.set.html\#20985}]
  $\Vz$ is closed under identity types.
\end{proposition}
\begin{proof}
  Let $x : \Vz$ and $a,a' : \Elz{x}$. Define the following map:
  \begin{align*}
    &f : a = a' → \Vz \\
    &f\,p := ∅
  \end{align*}
  This is an embedding as it is a map from an h-proposition into an h-set. The
  identity type in $\Vz$ is then defined as follows:
  \[
    \Idz{x}{a}{a'} := \supz\,(a = a', f) \qedhere
  \]
\end{proof}

The decoding holds up to definitional equality:
\[
  \Elz{(\Idz{x}{a}{a'})} ≡ \left(a = a'\right)
\]

We emphasize that $\Elz{x}$ is an h-set for any for any $x : \Vz$\!. Accordingly, $\Idz{x}{a}{a'}$ is
necessarily a proposition for any $a,a' : \Elz{x}$. In particular, $\Idz{x}{a}{a'}$ satisfies
UIP. As this identity type represents internalizes the ambient identity type, other expected
properties of the identity type (such as function extensionality) also hold.

\subsection{Set quotients}
\label{sec:vsetquotients}

In order to define set quotients in $\Vz$\!, we must assume that these quotients exist
in our starting universe $\U$. More specifically, we first assume that there is a
function of the following type
\[
  {-}/{-} : ∏_{A : \U} (A → A → \U) → \TThSetU
\]

We then ensure that $A / R$ realizes the quotient of $A$ by the relation $R$ by requiring a map
$[-]_R : A → A / R$ such that $R␣a␣b → [a]_R = [b]_R$ for all $a,b : A$. We also assume a suitable
elimination principle: given a family of h-sets $P : A / R → \TThSetU$, we can construct a function
$∏_{x : A / R} P\,x$ from a function $q : ∏_{x : A / R} P\,[a]_R$ which coheres with the map
$R␣a␣b → [a]_R = [b]_R$. The require that function we get satisfies a coherence condition, and if we
precompose it with the quotient map $[-]_R$ we get back $q$. (For the exact assumptions, see the
formalization
\agdalink{https://elisabeth.stenholm.one/category-of-iterative-sets/set-quotient.html}.)

While we don't need to assume that $R : A → A → \U$ is an equivalence relation (h-propositional,
symmetric, reflexive, transitive), the constructions below will use the fact any $R$ induces an
equivalence relation $|R| : A → A → \TThPropU$ defined by $|R|␣a␣b ≔ \left([a]_R = [b]_R\right)$.

To streamline the process, we will use an interesting formulation of equivalence relations:
\begin{lemma}\label{lma:equiv-rel}
Given a relation $R : A → A → \TThPropU$, the following are equivalent:

\begin{itemize}
 \item $R$ is an equivalence relation
 \item $R␣a␣b ≃ ∏_{c:A} \left(R␣a␣c ≃ R␣b␣c\right)$ for all $a,b:A$
 \item $R␣a␣b ≃ \left(R␣a =_{A → \U} R␣b\right)$ for all $a,b:A$
\end{itemize}
\end{lemma}

\begin{proof}
 Since $R$ is an (h-propositional) binary relation, the above
 statements are all h-propositions. The last two are equivalent by
 function extensionality and univalence. It thus remains to
 show that being an equivalence relation is equivalent to one of
 the last two--we choose the middle one.

 Assume that $R$ is an equivalence relation. Everything in sight is an
 h-proposition, so the equivalences are logical equivalences. Thus, assume that
 $R␣a␣b$. Then we get a map $∏_{c:A} R␣a␣c ↔ R␣b␣c$ by transitivity and
 symmetry. In the other direction, if $∏_{c:A} R␣a␣c ↔ R␣b␣c$, choose $c = a$ in
 order to obtain $R␣a␣a ↔ R␣a␣b$. Since $R$ is reflexive, we get $R␣a␣b$.

 Conversely, assume $R␣a␣b ≃ ∏_{c:A} \left(R␣a␣c ≃ R␣b␣c\right)$ for all
 $a,b:A$. To show reflexivity, let $b=a$ and notice that $∏_{c:A} \left(R␣a␣c ≃
 R␣a␣c\right)$ has a canonical element, from which we obtain $R␣a␣a$. Symmetry
 is obtained by observing that $ ∏_{c:A} \left(R␣a␣c ≃ R␣b␣c\right) ≃ ∏_{c:A}
 \left(R␣b␣c ≃ R␣a␣c\right)$ and hence $R␣a␣b ≃ R␣b␣a$. For transitivity,
 remember that $R␣a␣b$ gives $∏_{c:A} \left(R␣a␣c ≃ R␣b␣c\right)$, thus if we
 have $R␣b␣c$ we get $R␣a␣c$ by following the backwards direction of the
 equivalence.
\end{proof}

The property $R␣a␣b ≃ ∏_{c:A} \left(R␣a␣c ≃ R␣b␣c\right)$ for all $a,b:A$ essentially states that
the equivalence classes of $R$ behave well. Tangentially, we note that this requirement make sense
even when $R$ is a general binary family, not only of h-propositions. Thus, this property might make
for interesting future study.

\begin{proposition}[\agdalink{https://elisabeth.stenholm.one/category-of-iterative-sets/iterative.set.html\#23635}]
  $\Vz$ is closed under set quotients. That is, given $a:\Vz$ and $R : \Elz a → \Elz a → \U$
  there is $a /⁰ R : \Vz$ such that $\Elz (a /⁰ R ) ≡ (\Elz a) / R$.
\end{proposition}
\begin{proof}
  By Proposition \ref{prop:equiv-emb-int} it suffices to construct an embedding
  $\Elz a / R ↪ \Vz$\!. We define $f :~\Elz a → \Vz$ and prove that for any
  $x,x':\Elz a$ we have $\left([x]_R = [x']_R\right) ≃ \left(f␣x = f␣x'\right)$.
  By the elimination principle for set quotients, this will induce an embedding
  $\Elz a / R ↪ \Vz$\!.

  Thus, let $f␣x = \supzc{∑_{y:\Elz a}|R|␣x␣y}{\widetilde a ∘ \TTfst}$, and
  observe the chain of equivalences:

  \begin{align*}
    \left(f␣x = f␣x'\right)
      &≡  \left(\supzc{∑_{y:\Elz a}|R|␣x␣y}{\widetilde a ∘ \TTfst} = \supzc{∑_{y:\Elz a}|R|␣x'␣y}{\widetilde a ∘ \TTfst}\right) \\
      &≃ ∑_{α : \left(∑_{y:\Elz a}|R|␣x␣y\right) ≃ \left(∑_{y:\Elz a}|R|␣x'␣y\right)} \widetilde a ∘ \TTfst = \widetilde a ∘α∘\TTfst \\
      &≃ ∑_{α : \left(∑_{y:\Elz a}|R|␣x␣y\right) ≃ \left(∑_{y:\Elz a}|R|␣x'␣y\right)}  \TTfst = α∘\TTfst \\
      &≃ ∏_{y : \Elz a} |R|␣x␣y ≃ |R|␣x'␣y \\
      &≃ |R|␣x␣x' \\
      &≡ \left([x]_R = [x']_R\right)
  \end{align*}

  Note that we have used the characterization of equivalence relations given by
  Lemma \ref{lma:equiv-rel} in the next to last step.
\end{proof}

\subsection{Using the type formers}

Using the types and type formers in $\Vz$ we can construct new types. The
decoding of these composite types will then hold up to definitional equality.
For example, given elements $x,y : \Vz$\!, a map $f : \Elz{x} → \Elz{y}$ and $b :
\Elz{y}$ we can define the code for the fiber of $f$ over $b$ as
\[
  \fibz{f}{b} := \Sigmaz{x}{(λa.\Idz{y}{(f\,a)}{b})}
\]
After applying the decoding function, we get obtain the following definitional equality:
\[
  \Elz{\left(\fibz{f}{b}\right)} ≡ \fib{f}{b}
\]

\subsection{Universes}
\label{sec:vuniverses}

The flexible handling of hierarchies of universes is a key feature
of dependent type theory. It makes it easy to formalize higher order
concepts, and mathematical structures. Our universe construction
retains this ability, and in this subsection we demonstrate that
by observing that the types $\Vz_0, \Vz_1, ⋯ ,\Vz_{\ell}, ⋯$ form
a hierarchy of universes, where each universe occurs as a type
with a code in the next.

\begin{proposition}[\agdalink{https://elisabeth.stenholm.one/category-of-iterative-sets/iterative.set.html\#16171}]
  For any universe level $\ell$ there is a code $\Vzcode{\ell} : \Vz_{\ell +
  1}$ for $\Vz_{\ell}$ with the definitional equality $\Elz{\Vzcode{\ell}} ≡ \Vz_{\ell}$
\end{proposition}
\begin{proof}
  Given a universe level $\ell$, we need to construct an embedding $\Vz_{\ell}
  ↪ \Vz_{\ell^+}$. For this, we start by constructing an embedding
  $\Vinf_{\ell} ↪ \Vinf_{\ell^+}$. Thus define the map
  \begin{align*}
    &\varphi : \Vinf_{\ell} → \Vinf_{\ell^+} \\
    &\varphi\,(\supop\,A\,f) := \supop\,A\,(\varphi ∘ f)
  \end{align*}
  (Note that we are using cumulative universes in the ambient type theory, so $A : \U_{\ell^+}$ whenever $A
  : \U_{\ell}$.) To show that $\varphi$ is an embedding, let $\supop\,A\,f,
  \supop\,B\,g : \Vinf_{\ell}$ be arbitrary elements. We need to show that
  \[
    \ap\,\varphi : \supop\,A\,f = \supop\,B\,g → \supop\,A\,(\varphi ∘ f) = \supop\,B\,(\varphi ∘ g)
  \]
  is an equivalence. First, note that the following diagram commutes:
  \begin{center}
    \begin{tikzcd}
      ∑_{X : \U_{\ell}} X → \Vinf_{\ell} \arrow[d, "\supop"'] \arrow[rrr, "{λ(X,h).(X,\varphi ∘ h)}"] &  &  & ∑_{X : \U_{\ell^+}} X → \Vinf_{\ell^+} \arrow[d, "\supop"] \\
      \Vinf_{\ell} \arrow[rrr, "\varphi"']                                                                         &  &  & \Vinf_{\ell^+}
    \end{tikzcd}
  \end{center}
  The map $\supop$ is an equivalence, so $\varphi$ is an embedding if and only
  if the top map is an embedding. Thus we need to show the following to be an equivalence:
  \[
    \ap\,(λ(X,h).(X,\varphi ∘ h)) : (A,f) = (B,g) → (A,\varphi ∘ f) = (B,\varphi ∘ g)
  \]
  While we might hope to argue that this is a fiberwise embedding and therefore the total map is an
  embedding as well. Unfortunately, our induction hypothesis does not state that $\varphi$ is an
  embedding, i.e., that $\ap$ is an equivalence for \emph{all} elements. It only ensures that it is
  an equivalence for \emph{some} elements. We must then take a different path and instead note
  that the diagram below commutes.
  \begin{center}
    \begin{tikzcd}
      {(A,f)=(B,g)} \arrow[d, "{≃}"'] \arrow[rrr, "{\ap\,(λ(X,h).(X,\varphi ∘ h))}"] &  &  & {(A,\varphi ∘ f) = (B, \varphi ∘ g)} \arrow[d, "≃"] \\
      ∑_{e : A ≃ B} f ∼ g ∘ e \arrow[rrr, "{λ(e,H).(e,\ap\,\varphi ∘ H)}"']      &  &  & ∑_{e : A ≃ B} \varphi ∘ f ∼ \varphi ∘ g ∘ e
    \end{tikzcd}
  \end{center}
  The vertical maps are provided by Theorem \ref{thm:Vinf-eq}. Using 3-for-2, the top map is an
  equivalence if and only if the bottom one is an equivalence. We now note that it suffices to check
  this property on fibers, so we need to show that given $e : A ≃ B$, the following map is an
  equivalence:
  \[
    (λH.\ap\,\varphi ∘ H) : f ∼ g ∘ e → \varphi ∘ f ∼ \varphi ∘ g ∘ e
  \]
  We now recall that postcomposition by a family of maps is an equivalence if it is a family of
  equivalences. Finally, we must argue that
  $\ap\,\varphi : f\,a = g\,(e\,a) → \varphi\,(f\,a) = \varphi\,(g\,(e\,a))$ is an equivalence for
  all $a : A$. This follows from the induction hypothesis. Thus, we conclude that
  $\varphi : \Vinf_{\ell} → \Vinf_{\ell^+}$ is an embedding.

  To argue that this equivalence restricts to $\Vz$\!, we must show this equivalence sends iterative
  sets to iterative sets.  Thus let $\supop\,A\,f : \Vinf_{\ell}$ be such that $f$ is an embedding
  and $f\,a$ is an iterative set for all $a : A$. We must argue that $\supop\,A\,(\varphi ∘ f)$ is
  an iterative set. But the map $\varphi ∘ f$ is an embedding as the composition of two
  embeddings. Moreover, by the induction hypothesis, for any $a : A$, $\varphi\,(f\,a)$ is an
  iterative set since $f\,a$ is an iterative set.

  Therefore, $\varphi$ is an embedding from $\Vz_{\ell}$ into $\Vz_{\ell^+}$.
  The code for $\Vz_{\ell}$ in $\Vz_{\ell^+}$ is thus defined as
  \[
    \Vzcode{\ell} := \supz\,(\Vz_{\ell}, \varphi) \qedhere
  \]
\end{proof}

Note that the decoding holds up to definitional equality:
\[
  \Elz{\Vzcode{\ell}} ≡ \Vz_{\ell}
\]

\begin{proposition}
  $\Vz$ is \emph{not} a univalent universe.
\end{proposition}
\begin{proof}
  For $x, y : \Vz$\!, $x = y$ is an h-proposition as $\Vz$ is an h-set,
  but $\Elz\,x ≃ \Elz\,y$ is in general a proper h-set.
\end{proof}

\section{\texorpdfstring{$\Vz$}{V0} as a category}
\label{sec:vcategory}

The universe structure on $\Vz$ induces a category with a full and faithful functor into
$\hSetcatU$. In this section, we define this category and show that it is closed under many
essential constructions (finite limits and colimits, exponentials, and more). This category provides
another concrete way for us to \emph{measure} the adequacy of $\Vz$ as a replacement for $\TThSetU$;
the former induces a closely related category to the latter, sharing many similar properties.

\begin{definition}[\agdalink{https://elisabeth.stenholm.one/category-of-iterative-sets/iterative.set.category.html\#428}]
  Let $\Vcat$ be the category with
  \begin{itemize}
  \item $\Ob\,\Vcat := \Vz$
  \item $\Hom_{\Vcat}\,(x,y) := \Elz{x} \to \Elz{y}$
  \item $\id$ and $\circ$ are simply the identity function and
    function composition.
  \end{itemize}

  All laws hold by $\refl$ as $\id$ and $\circ$ are the identity and composition
  from the ambient type theory. For $x,y : \Vz$\!, the type $\Hom\,(x,y)$ is an
  h-set as it consists of functions into an h-set.
\end{definition}

Note that we will take \emph{category} to denote what the
\citetalias{HoTT13} calls ``precategory'', \emph{univalent category} to
denote what the book calls ``category'', and \emph{strict
  category} to denote a category where the objects form an h-set.
Hence, $\Vcat$ is a strict category as $\Vz$ is an
h-set.

The following holds more-or-less by construction:
\begin{lemma}[\agdalink{https://elisabeth.stenholm.one/category-of-iterative-sets/iterative.set.category.html\#906}]
  \label{lem:elfff}
  The map $\Elzop$ induces a full and faithful functor from $\Vcat$ to $\hSetcatU$.
\end{lemma}

Clearly, $\Vcat$ is not a univalent category since it possesses objects with
non-trivial automorphisms, but the type of objects in $\Vcat$ is an h-set.
Still, one might ask if $\Elzop$ is an equivalence of categories. This
does not appear to be true in general, but can be implied by further axioms.
For instance, the axiom of choice implies that $\Elzop$ is an equivalence.
The core of this is whether every type in $\U$ can be equipped with an
iterative set-structure---a property known as well-founded materialization.
We discuss this further in Section~\ref{sec:othervs:hset}.

Fortunately, even without additional axioms we are able to show that $\Vcat$
retains much of the essential structure of $\hSetcatU$ and that $\Elzop$
preserves many important categorical constructions. In order to show that
$\Vcat$ is closed under some categorical structure, it therefore suffices to
break the process into two stages:
\begin{enumerate}
\item Show that $\hSetcatU$ is closed under \emph{e.g.,} finite limits, exponentials, \emph{etc.}
\item Show that the objects involved land in the image of $\Elzop$\!.
\end{enumerate}

Better still, $\hSetcatU$ is well-studied and known to be closed under all the
categorical structures we consider \citep{RS15}. Our task is
therefore reduced only to showing that various objects of $\hSetcatU$ land in
the image of $\Elzop$\!. For this, we repeatedly capitalize on the fact that the
decoding $\Elzop$ holds up to definitional equality; it ensures that the final
step can be rephrased as follows: show that there exists an iterative structure
on the objects involved.
This pattern is used repeatedly to prove the following result:

\begin{therm}[\agdalink{https://elisabeth.stenholm.one/category-of-iterative-sets/iterative.set.category.properties.html}]
  $\Vcat$ is closed under and $\Elzop$ preserves the following:
  \begin{enumerate}
  \item initial object,
  \item terminal object,
  \item finite coproducts,
  \item pushouts,
  \item finite products,
  \item pullbacks, and
  \item exponentials.
  \end{enumerate}
\end{therm}
\begin{proof}
  As $\hSetcatU$ supports these structures it suffices to show that each of the
  representing objects land in the image of $\Elz$\!. This clearly follows from the results in
  Section~\ref{sec:vtarski}, for instance, the existence of the initial and terminal object follows
  from Proposition~\ref{prop:empty-unit-booleans}, and e.g., pullbacks can be constructed through $\Sigmaz$ and $\fibopz$ just like in $\hSetcatU$.
\end{proof}

As $\Vcat$ has pullbacks/pushouts and terminal/initial object we directly get:

\begin{corollary}
  $\Vcat$ has finite limits and colimits. These are preserved by $\Elzop$\!.
\end{corollary}

We defer further categorical considerations of $\Vcat$ to future work
and instead turn our attention to slice categories of $\Vcat$, which
play an important role in the study of it as a model of type theory.

\subsection{Slice categories of $\Vcat$}

Similar methods to the ones above also apply when showing that the slice categories
$\Vcat/a$ are well-behaved. In particular, $\Elz$ induces a full and faithful
functor $\Vcat/a \to \hSetcatU/\Elz{a}$. We can use this fact to deduce,
e.g., that $\Vcat/a$ is cartesian closed.

\begin{proposition}[\agdalink{https://elisabeth.stenholm.one/category-of-iterative-sets/iterative.set.category.slices.properties.html}]
  For any $a : \Vz$\!, $\Vcat / a$ has finite limits.
\end{proposition}
\begin{proof}
  This can be proved using the standard argument: products in a slice category
  are realized by pullbacks in the underlying category and connected limits are
  realized by limits of the underlying diagram. We could also argue by noting that
  $\hSetcatU/\Elz{a}$ is finitely complete and that limits of diagrams in the
  image of $\Elz$ remain in the image of $\Elz$\!.
\end{proof}

We give a bit more details in the following proof as it showcases the usefulness
of being able to encode things directly in $\Vz$\!, combined with the fact that
$\Elz$ strictly decodes to the expected thing in $\U$.

\begin{proposition}[\agdalink{https://elisabeth.stenholm.one/category-of-iterative-sets/iterative.set.category.slices.properties.html\#7694}]
  For any $a : \Vz$\!, $\Vcat / a$ has exponentials.
\end{proposition}
\begin{proof}
  Given $(x,f), (y,g) : \Ob\,(\Vcat / a)$, define their exponential as the
  element
  \[
    \mathsf{exp}\,(x,f)\,(y,g) := \Sigmaz{a}{(\lambda\,i.\,\funz{\fibz{g}{i}}{\fibz{f}{i}})}
  \]
  Note that we have the following definitional equality:
  \[
    \Elz{(\mathsf{exp}\,(x,f)\,(y,g))} ≡ \sum_{i : \Elz{a}} \fib{g}{i} → \fib{f}{i}
  \]
  This is the exponential in $\hSetcatU/\Elz{a}$, so $\mathsf{exp}\,(x,f)\,(y,g)$ is an
  exponential object in $\Vcat/a$.
\end{proof}

\begin{corollary}
  $\Vcat$ is locally cartesian closed and $\Elzop$ is a locally cartesian closed functor.
\end{corollary}

Finally, the following proposition foreshadows the next section where
we build a model of type theory on $\Vcat$. In that section, we wish to
interpret types in context $a$ as elements of $\Vcat/a$ and to realize
substitution as pullback. It is well-known, however, that the result is merely
pseudofunctorial and therefore insufficient to form a (strict) model of type
theory~\citep{Seely84,CurienGarnerHofmann14}. In the specific case of
$\hSetcatU$, there is a well-known pseudo-natural equivalence between the slice
category $\hSetcatU/a$ and the functor category $[a,\hSetcatU]$ which remedies
the coherence issues. This equivalence restricts to the full subcategory
determined by $\Vcat$:

\begin{proposition}[\agdalink{https://elisabeth.stenholm.one/category-of-iterative-sets/iterative.set.category.slices.functor.html\#8836}]
  \label{prop:slice-equiv}
  Given $a : \Vz$ and writing $a$ for the corresponding discrete category
  associated with $\Elz{a}$, there is a canonical equivalence $\Vcat / a \simeq
  [a,\Vcat]$.
\end{proposition}
\begin{proof}
  The equivalence is constructed in the standard way. The functor from
  $[a,\Vcat]$ to $\Vcat / a$ sends $F : \Ob\,[a,\Vcat]$ to the element
  $\Sigmaz{a}{F}$ together with the first projection. In the other direction, an
  element $(x,f) : \Ob\,(\Vcat / a)$ is mapped to the functor $\lambda i.
  \fibz{f}{i}$.
\end{proof}

\section{\texorpdfstring{$\Vz$}{V0} as a category with families}
\label{sec:vcwf}

Having established that $\Vz$ organizes into a
well-behaved category, we now take this a step further by showing that $\Vz$
supports a model of extensional type theory. Since our goal is to do this very
formally, our task is threefold: first, we must specify what we mean by a model
of type theory. To this end, we have formalized a particular presentation of a
\emph{category with families (CwF)} \citep{Dybjer96}. This extends a category with the additional
structure required to interpret dependent type theory. Next, we show that
$\Vcat$ can be equipped with this additional structure. Finally, since our
definition of a CwF does not prescribe closure under any
connectives, we detail how to extend a CwF with various connectives and show
that the CwF structure on $\Vcat$ supports these extensions.

\begin{remark}\label{remark:why-sig-structure-sucks}
  Of these three steps, only the first two are fully formalized. The obstruction
  to formalizing closure under all relevant extensions is, surprisingly,
  completely independent of $\Vz$\!. Rather, it stems from the fact that the
  equations governing substitution hold only up to propositional equality,
  leading to complicated path and transport computations when defining the
  substitution properties of said structure.
\end{remark}

\newcommand{\C}{\mathcal{C}}
\newcommand{\Ty}{\mathsf{Ty}}
\newcommand{\Tm}{\mathsf{Tm}}
\newcommand{\Op}[1]{{#1}^{\mathsf{op}}}
\newcommand{\Elcat}[1]{\int\mkern -3mu #1}

\subsection{The definition of categories with families}

In the paper and formalization we rely on the following formulation of categories with families.

\begin{definition}[Category with families, \agdalink{https://elisabeth.stenholm.one/category-of-iterative-sets/type-theories.precategories-with-families.html\#2067}]
  A category with families (CwF) consists of:
  \begin{itemize}
    \item A category $\C$ with a terminal object,
    \item a presheaf $\Ty : \Op{\C} → \hSetcatU$,
    \item a presheaf $\Tm : \Op{\left(\Elcat{\Ty}\right)} → \hSetcatU$,
    \item a functor ${-}.{-} : \Elcat{\Ty} → \C$, and
    \item for each $Δ : \Ob\,\C$ and $(Γ , A) : \Elcat{\Ty}$ a natural
    equivalence
    \[\Hom\,(Δ, Γ.\,A) ≃ ∑_{γ\;: \Hom\,(Δ,Γ)} \Tm\,(Δ, A \cdot γ)\]
  \end{itemize}
\end{definition}

Here we have written $A \cdot γ$ for $\Ty\,(γ\,A)$, and $\Elcat{\Ty}$ for the
\emph{category of elements} of $\Ty$, \ie{}, the total space of the right
fibration induced by $\Ty$.
Intuitively, objects in $\C$ interpret the contexts of our type theory,
while morphisms interpret substitutions. The additional presheaves
are used to interpret types and terms. Specifically, the set of semantic types
in context $Γ : \Ob\,\C$ is given by $\Ty\,Γ$ while the set of terms of type
$A : \Ty\,Γ$ is given by $\Tm\,(Γ,A)$.  The functoriality of $\Ty$ and $\Tm$ is
precisely the structure required to interpret the application of substitutions
to types and terms.

The terminal object interprets the empty context and the functor from
$\Elcat{\Ty}$ to $\C$ interprets context extension. A context $Γ: \Ob\,\C$ can be
extended by a type $A : \Ty\,Γ$ in that context, to produce a new context $Γ.\,A
: \Ob\,\C$.

Finally, the natural equivalence ties together substitutions and elements of
$\Tm$. In particular, the inverse encodes the ability to extend a substitution
with a term. Following this last observation, we write $Γ.\,a$ for the element
$\Hom\,(Δ,Γ.\,A)$ induced by the element $(Γ,a) : \sum_{γ\;: \Hom\,(Δ,Γ)}
\Tm\,(Δ, A \cdot~γ)$.

\subsection{Equipping $\Vcat$ with a CwF structure}

We now turn to equipping $\Vcat$ with a CwF structure. We begin by defining
$\Ty$ as follows:
\[
  \Ty\,X := \Elz{X} \to \Vz
\]

Intuitively, a type in context $X$ is precisely an $X$-indexed family of sets.
There is, however, a major subtlety in this definition that should be
emphasized: the version of the definition where $\Vz$ is replaced by
$\TThSetU$ would be \emph{incorrect}. We have required that $\Ty\,X$ always be
an h-set as it is assumed to be an $\TThSetU$ valued presheaf. Therefore, it is
only after finding an adequate ``h-set of h-sets'' that we can define the set
model of type theory in this manner.

The definition of the presheaf of terms is also reasonably direct:
\[
  \Tm\,(X, A) = \prod_{x : \Elz{X}} \Elz{(A\,x)}
\]

We now show that, along with $\Vcat$, these two definitions assemble into a CwF.
\begin{proposition}[\agdalink{https://elisabeth.stenholm.one/category-of-iterative-sets/iterative.set.cwf-structure.html\#3203}]
  $\Vcat$ can be equipped with a CwF structure.
\end{proposition}
\begin{proof}
  We have given the putative definitions of $\Ty$ and $\Tm$. We note that it is
  straightforward to ensure that both are suitably functorial. The functorial
  action is given by precomposition and all the required equations hold
  on-the-nose.

  It remains to show that these three pieces of data satisfy the required
  properties of a CwF.  We have already shown that $\Vcat$ has a terminal
  object, so it remains to discuss the interpretation of context extension.
  Fix $X : \Vz$ and $A : \Ty\,X$. We define $X.\,A : \Vz$ as $\Sigmaz{X}{A}$.
  The natural equivalence then follows from the $\eta$ principle of dependent
  sums.
\end{proof}
By virtue of Proposition~\ref{prop:slice-equiv}, we further note that types $A$ in context $X$ in
this model are realized up to equivalence by families $\Elz{A} \to \Elz{X}$ and terms are likewise
determined by sections. By presenting $\Ty$ and $\Tm$ in terms of (dependent) products rather than
families and sections, we are able to equip both with strictly functorial actions.

We emphasize that the accomplishment here is not in the definition itself; it mirrors the na{\"i}ve
definition of the set model of type theory as presented by e.g., \citet{Hofmann97}.  What is crucial
is that $\Vz$ retains enough of the good behavior of $\TThSetU$ to support such a straightforward
definition of the CwF structure while still managing to be an h-set itself.

\begin{remark}
  We note that there are many closely related presentations of models of type theory (categories
  with attributes~\citep{Cartmell86}, contextual categories~\citep{Cartmell86,Streicher91},
  comprehension categories~\citep{Jacobs93,Jacobs99}, natural models~\citep{Fiore12,Awodey18natural}
  and so on). We have opted for CwFs because the CwF structure on $\Vz$ is particularly simple and
  enjoys an exceptional number of definitional equalities. In particular, as opposed to other models
  which recover terms indirectly as sections to display maps, CwFs require the presheaf of terms as
  part of their data. This allows us to choose a particular definitional representative for the type
  of terms in our model and we are then able to explicitly select dependent functions. We shall see
  that this makes closing $(\Vcat, \Ty, \Tm)$ under various constructions particularly
  straightforward, as most naturality conditions hold definitionally.
\end{remark}

\subsection{Further structure on $\Vz$ as a CwF}

While we have constructed a CwF structure on $\Vcat$, we have thus far only shown that the model
interprets the basic structural rules of type theory, but not that it is closed under any connectives. The
process of extending the model with new connectives is essentially modular: for each connective, we
specify the relevant structure on top of a CwF necessary to interpret it and then show that the CwF
$(\Vcat,\Ty,\Tm)$ supports this additional structure.

We illustrate the process with $\Pi$-types. First, we must define a $\Pi$-structure on a CwF.
\newcommand{\PiTy}{\mathsf{pi}}
\newcommand{\TmIso}{\alpha}
\begin{definition}[$\Pi$-structure, \agdalink{https://elisabeth.stenholm.one/category-of-iterative-sets/type-theories.pi-types-precategories-with-families.html\#869}]
  A $\Pi$-structure on a CwF $(\C,\Ty,\Tm)$ is defined by the following:
  \begin{itemize}
  \item An operation $\PiTy : ∏_{Γ : \Ob\,\C} \Ty\,Γ → \Ty\,(Γ.\,A) → \Ty\,Γ$, natural in $Γ$\!.
  \item For any $Γ : \Ob\,\C$, $A : \Ty\,Γ$\!, and $B : \Ty\,(Γ.\,A)$ an isomorphism $\TmIso_\PiTy$
    between $\Tm\,Γ\,(\PiTy\,Γ\,A\,B)$ and $\Tm\,(Γ.A)\,B$, natural in $Γ$\!.
  \end{itemize}
\end{definition}

\begin{lemma}[\agdalink{https://elisabeth.stenholm.one/category-of-iterative-sets/iterative.set.cwf-structure.html\#3941}]
  The CwF $(\Vcat,\Ty,\Tm)$ supports a $\Pi$-structure.
\end{lemma}
\begin{proof}
  We begin by defining $\PiTy$ as follows:
  \[
    \PiTy\,X\,A\,B := \lambda (x : \Elz(X)).\, \Piz{(A\,x)}{(\lambda a.B\,(x,a))}
  \]
  Naturality in $X$ is a straightforward computation. The definition of $\TmIso_{\PiTy}$, after
  unfolding, reduces to the manifestly natural equivalence induced by currying:
  \[
    \prod_{x : X}\prod_{y : Y(x)} Z(x,y) \simeq \prod_{p : \sum_{x : X} Y(x)} Z(p) \qedhere
  \]
\end{proof}

We have formalized both the definition of $\Pi$-structures and the particular $\Pi$-structure on
$(\Vcat,\Ty,\Tm)$ in \Agda. However, already some small inconveniences emerge. For
instance, in the definition of naturality for $\TmIso_{\PiTy}$, we must specify an equality
dependent on the proof witnessing naturality of $\PiTy$. The dependence is straightforward in this case,
but becomes more complex for the later structures. Accordingly, we present only paper proofs for
them.

Furthermore, $(\Vcat,\Ty,\Tm)$ also supports dependent sums.
\newcommand{\SigmaTy}{\mathsf{sig}}
\begin{definition}[$\Sigma$-structure]
  A $\Sigma$-structure on a CwF $(\C,\Ty,\Tm)$ consists of the following two pieces of
  data:
  \begin{itemize}
  \item An operation $\SigmaTy : ∏_{Γ : \Ob\,\C} \Ty\,Γ → \Ty\,(Γ.\,A) → \Ty\,Γ$\!, natural in $Γ$\!.
  \item For any $Γ : \Ob\,\C$, $A : \Ty\,Γ$\!, and $B : \Ty\,(Γ.\,A)$ a natural isomorphism $\TmIso_\SigmaTy$
    between $\Tm\,Γ\,(\SigmaTy\,Γ\,A\,B)$ and pairs $\sum_{a : \Tm(Γ,A)} \Tm\,Γ,(B \cdot (\id.a))$.
  \end{itemize}
\end{definition}

\begin{lemma}
  The CwF $(\Vcat,\Ty,\Tm)$ supports a $\Sigma$-structure.
\end{lemma}
\begin{proof}
  We define $\SigmaTy$ as follows:
  \[
    \SigmaTy\,X\,A\,B := \lambda (x : \Elz{X}).\, \Sigmaz{(A\,x)}{(\lambda a.B\,(x,a))}
  \]
  The remaining structure follows directly. In particular, even though the naturality requires
  complex path algebra to state properly in the specific CwF on $\Vcat$ all these paths are given by
  reflexivity.
\end{proof}

By similar considerations, we may define and close $(\Vcat,\Ty,\Tm)$ under many
other connectives: extensional identity types, booleans, natural numbers, and universes among
others. Putting all of this together, we conclude the following:

\begin{therm} \label{thm:vettmodel}
  $\Vcat$ supports a model of extensional type theory with the standard connectives.
\end{therm}

We note that this result, combined with Proposition~\ref{prop:slice-equiv} and the series of
results about $\Elzop$ preserving various categorical connectives can be summarized by the
informal slogan: $\Vcat$ supports a model of type theory which internalizes the set-level
fragment of the ambient type theory.

\section{Relationship to set models of type theory and other set universes in HoTT/UF}
\label{sec:otherv}

The idea of set theoretic semantics of type theory is of course an old
and natural one. An early reference where this is written down more
formally is the master thesis of \citet[Chapter~5]{Salvesen84}. As
discussed in the introduction the work presented in this paper goes
back to the model of CZF in type theory of \citet{Aczel78}. Aczel also
interpreted extensional type theory with universes in an extension of
CZF with a hierarchy of inaccessible sets \citep{Aczel99}. In fact,
Aczel's V occurs already in the PhD thesis of \citet{Leversha76} where
it was used to represent ordinals constructively.
Various earlier work has also relied on Aczel's V to model type
theory. For instance, \citet{Werner97} modeled the core system of
\Coq in ZFC and vice versa, using Aczel's encoding of sets. A
refinement by \citet{Barras10,BarrasHabil} models the core system of
\Coq system in intuitionistic ZF, and formalizes the model in \Coq
\citep{Coq}.  More recently, \citet{palmgren19} presented an
interpretation of extensional Martin-Löf type theory \citep{HAN} into
intensional Martin-Löf type theory via setoids, also relying on
Aczel's V. Palmgren's work was also formalized in \Agda.

Aczel's V was revisited in HoTT/UF by \citet{Gylterud18,Gylterud20}
who observed that this gives a universe of multisets, but that one can
restrict it, as in Definition~\ref{def:v}, to get a universe of
h-sets. These universes of (multi)sets has recently also been further
studied by Escardó and de Jong who has their own \Agda formalization
as part of the \systemname{TypeTopology} project
\citep{typetop}. Among many other things, they have two more proofs of
Theorem~\ref{thm:Vzhset} formalized.
Various HITs for representing \emph{finite} multisets have also been
considered in HoTT/UF
\citep{BasoldGeuversVanDerWeide17,FruminGeuvers+18,ChoudhuryFiore19,internalizing,veltri,veltri2},
however these are of course not sufficient to model full type theory.

We will now discuss other approaches to constructing strict categories
of sets in HoTT/UF that could also serve as internal models of type
theory. These often require various extensions of the quite minimal
univalent type theory that we have relied on in this paper.

\subsection{The cumulative hierarchy in the HoTT Book}

The HoTT Book postulates a universe of sets as a higher inductive type
called the \emph{cumulative hierarchy}
\citepalias[Definition~10.5.1]{HoTT13}. \citet[Section 8]{Gylterud18}
establishes an equivalence between the HoTT Book $V$ and $\Vz$\!, which
makes it possible to transfer all of our results over to $V$\!. One
remark about the HoTT Book $V$ is that it is h-set truncated, while
$\Vz$ is not. This means that the eliminator one gets for the HoTT
Book $V$ only lets one directly eliminate into h-sets, while $\Vz$ can
be directly eliminated into types of arbitrary homotopy level.
Similarly many basic constructions, like $∈\,:\,V → V → \U$, is a bit
more complicated to define for the HoTT Book $V$ as it is not
sufficient to only define them for point constructors, but one has to
check that the definitions are compatible with the higher constructors
as well. A practical and appealing aspect of $\Vz$ is hence that it is
easy to define operations by pattern-matching on it. Another is that
it is not postulated, but simply constructed from W-types.

\subsection{Inductive-recursive universes}

An alternative approach to modeling type theory in type theory is to
rely on quotient inductive-inductive types (QIITs) as considered by
\citet{AK16}. However, they run into the same problem as discussed
above when working in HoTT and trying to eliminate their QIIT into
$\TThSetU$. In particular, as the QIIT representation of type theory is
h-set truncated they cannot eliminate directly into $\TThSetU$ as it
is a 1-type (the same issue also applies to the HoTT Book $V$). The
authors resolve this by considering an inductive-recursive universe
closed under the relevant structure, which can be shown to be a set
without any need to set truncate. This enjoys many of the nice
properties of $\Vz$\!, like $\El$ decoding type constructors
definitionally, but induction-recursion is proof theoretically quite
strong and it is again interesting to emphasize that we can construct
$\Vz$ using only W-types.

\subsection{Covered Marked Extensional Well-founded Orders (MEWOs)}

In their recent paper, \cite{dejong_kraus_forsberg_xu_2023} show that
the HoTT Book $V$ is equivalent to the type of covered marked
extensional well-founded orders ($\MEWOcov$), and hence to $\Vz$\!. The
results in this paper thus imply that the type $\MEWOcov$ can be
equipped with a universe structure. A strength of the universe $\Vz$
is the computational aspect of the decoding function $\Elz$\!.
Unfortunately, the two underlying maps of the equivalence between
$\Vz$ and $\MEWOcov$ do not compose definitionally to the identity
when going from $\Vz$ to $\MEWOcov$ and back again. This means that
the induced decoding for $\MEWOcov$ given by going to $\Vz$ and then
applying $\Elzop$ is not as computationally well-behaved as $\Elzop$
on $\Vz$\!, as the decoding will only hold up to propositional equality.

\subsection{Relationship to $\hSetcatU$}
\label{sec:othervs:hset}

One reason to consider the category of iterative sets is to regard it
as a replacement for $\hSetcatU$. As noted in
Section~\ref{sec:vcategory}, $\Elzop$ induces a fully-faithful
functor, but it may fail to be essentially surjective. The statement
that $\Elzop$ is essentially surjective corresponds to Shulman's
\emph{axiom of well-founded materialization}
\citep{shulman_stack_2010} and which is, in turn, implied by the axiom
of choice.

If the functor is essentially surjective, it forms a categorical
equivalence between $\Vcat$ and a univalent category and thus
describes $\hSetcatU$ as the \emph{Rezk completion}~\citep{AKS15}
of~$\Vcat$. Informally, this shows, modulo classical axioms, that
$\Vcat$ is a more rigid presentation of $\hSetcatU$. Moreover, even
without additional axioms $\Vcat$ and $\Vz$ are closed under
essentially every construction of interest.

\subsection{Well-ordered sets}

Another approach to defining a strict universe of sets, inspired by
Voevodsky's simplicial set model \citep{KapulkinLumsdaine12}, is to
consider well-ordered sets. By relying heavily on Zermelo's
well-ordering principle, and hence choice, one can obtain a strict
category of well-ordered sets with the relevant structure, also as a
subcategory of $\hSetcatU$. This was experimented with in
\systemname{UniMath} \citep{UniMath} by \citet{woset}. However, this
turned out to be harder to work with formally than expected because of
all the propositional truncations and hence not completed.
Furthermore, if completed this would only \emph{merely} give us the
existence of an internal model and hence lead to a weaker result than
Theorem~\ref{thm:vettmodel}.

\section{Conclusion and future work}
\label{sec:conclusion}

We have constructed a universe of h-sets that is itself an h-set and structured
it into an internal model of extensional type theory. This main result can
perhaps also be proven for other h-set universes of h-sets, such as the ones
mentioned in Section \ref{sec:otherv}, but certain properties of our
construction makes it very convenient to work with, also formally. First and
foremost, the definitional decoding of type formers means that one avoids
complex transports. Secondly, the construction is carried out using basic
type-formers, and has a (provable) elimination principle which directly allows elimination
into general types. This development works in a fairly minimalist univalent
type theory, as long as it has W-types. These W-types can be large, and need
only be small if one wants to reflect a hierarchy of universes, as in Section
\ref{sec:vuniverses}. The results should thus have a broad applicability in
models of~HoTT/UF.

In the formalization, we stopped short of adding additional structure
to the CwF on $\Vcat$ after Π-types. The obstacles are in fact not in
providing the structure for our model, such as Σ-structure, but the
general formulation of what that extra structure constitutes on CwFs
based on categories (see Remark
\ref{remark:why-sig-structure-sucks} in Section \ref{sec:vcwf}). To
the best of our knowledge there are no other formalizations of CwFs
with $\Sigma$-structure out there that do not assume UIP or other
axioms and which do not use setoids or a more extensional equality. It
would be interesting to attempt formalizing this in cubical type
theory \citep{CCHM18} where equality of $\Sigma$-types is easier to
work with because of the primitive path-over types in the form of
$\mathsf{PathP}$-types. An experiment along these lines was performed
by \citet{vezzosicwf} in \CubicalAgda \citep{cubicalagda2}.  In this
small formalization Vezzosi considered the CwF structure on h-set
valued presheaves. It would of course not have been possible to fully
complete this for the same reason as discussed in this paper, but it
turned out that some of the constructions and equations that one has
to check were easier than in a corresponding formalization in \UniMath
by \citet{presheafunimath}. This also suggests a further direction to
explore: $\Vcat$ valued presheaves. These should enjoy the same nice
properties as $\hSetcatU$ valued presheaves, but it should be
possible to organize also them into a model of type theory internally
in HoTT/UF.

Another avenue of further study is to take a closer look at Shulman's axiom of
well-founded materialization. Just like univalence, it makes sense to formulate
this axiom relative to a given universe of types. The construction of $\Vz$ can
be carried out on any universe, so a reasonable reformulation of well-founded
materialization in type theory could be: a universe $\U$ has well-founded
materialization if $\Elz : \Vz_\U → \hSetcatU$ is essentially surjective.
As mentioned, this follows from AC, but does not seem to be inherently
non-constructive. For instance, $\Vz$ itself has well-founded materialization
for trivial reasons. The most pertinent question is perhaps whether well-founded
materialization and univalence can constructively coexist. If we start with a
univalent $\U$, one could take the image of $\Elz$ in $\hSetcatU$ to
obtain a univalent universe which also somewhat trivially has well-founded
materialization. However, it is not immediate that this is closed under Π-types
and Σ-types as a naïve attempt quickly runs into choice problems.
 
\bibliographystyle{msclike}
\bibliography{refs.bib}

\end{document}